\pdfoutput=1
\documentclass[a4paper]{article}
\usepackage[utf8]{inputenc}
\usepackage[T1]{fontenc}

\usepackage[dvipsnames]{xcolor}
\usepackage{amsmath,amsthm,amssymb}
\usepackage{authblk}
\usepackage{doi}
\usepackage{enumitem}
\usepackage{lmodern}
\usepackage{mathtools}
\usepackage{hyperref}
\usepackage{listings}
\usepackage{microtype}
\usepackage[square, numbers]{natbib}
\usepackage{orcidlink}
\usepackage{tikz,pgfplots}

\pgfplotsset{compat=1.18}

\newcommand{\N}{\mathbb{N}}
\newcommand{\Z}{\mathbb{Z}}
\newcommand{\R}{\mathbb{R}}
\newcommand{\Rhat}{\widehat{R}}
\DeclareMathOperator{\sign}{sign}
\DeclareMathOperator{\Mod}{mod}

\newcommand{\Signs}{\Sigma}

\newtheorem{theorem}{Theorem}
\newtheorem{corollary}[theorem]{Corollary}
\newtheorem{proposition}[theorem]{Proposition}
\theoremstyle{definition}
\newtheorem{lemma}[theorem]{Lemma}
\theoremstyle{remark}
\newtheorem{example}[theorem]{Example}

\begin{document}

\title{On the Number of Real Types of\\ Univariate Polynomials}

\author[1]{Nicolas Faroß\,\orcidlink{0009-0001-4419-2337}\,}
\author[2,3]{Thomas Sturm\,\orcidlink{0000-0002-8088-340X}\,}
\affil[1]{Department of Mathematics, Saarland University, Saarbrücken, Germany\authorcr
\href{mailto:faross@math.uni-sb.de}{faross@math.uni-sb.de}}

\affil[2]{CNRS, Inria, and the University of Lorraine, Nancy, France\authorcr
\href{mailto:thomas.sturm@cnrs.fr}{thomas.sturm@cnrs.fr}}

\affil[3]{MPI-INF and Saarland University, SIC, Saarbrücken, Germany\authorcr
\href{mailto:sturm@mpi-inf.mpg.de}{sturm@mpi-inf.mpg.de}}

\maketitle

\begin{abstract}
The real type of a finite family of univariate polynomials characterizes the
combined sign behavior of the polynomials over the real line. We derive an
explicit formula for the number of real types subject to given degree bounds.
For the special case of a single polynomial we present a closed-form
expression involving Fibonacci numbers. This allows us to precisely describe
the asymptotic growth of the number of real types as the degree increases,
in terms of the golden ratio.
\end{abstract}

\section{Introduction}

We consider univariate polynomials with real coefficients. The \emph{real
type} of such a polynomial is the finite sequence of its signs listed from
$-\infty$ to $\infty$. For instance, the polynomial $f = x^4 - 4x^2$ plotted
in Figure~\ref{fig:intro} has real roots at $-2$, $0$, $2$. Evaluation of the
signs of $f$ at and between those real roots yields its real type
\begin{equation}
    \label{eq:real_type_1}
    [1, 0, -1, 0, -1, 0, 1].
\end{equation}

The notion of a real type generalizes to finite families $f_1$,
\dots,~$f_n$ of polynomials by considering their combined signs as column
vectors of length $n$. For instance, all real roots of the polynomials of the
family $f_1 = x + 1$, $f_2 = 2x + 1$, $f_3 = x^2 - 1$ plotted in
Figure~\ref{fig:intro} are located at $-1$, $-\frac12$, $1$. Evaluation of the
signs of  $[f_1, f_2, f_3]^T$ at and between those points yields the real type
\begin{equation}
    \label{eq:real_type_m}
    \begin{bmatrix*}[r]
    -1 & 0  & 1  & 1  & 1  & \phantom{-}1 & \phantom{-}1 \\
    -1 & -1 & -1 & 0  & 1  & 1 & 1 \\
    1  & 0  & -1 & -1 & -1 & 0 & 1
    \end{bmatrix*}.
\end{equation}

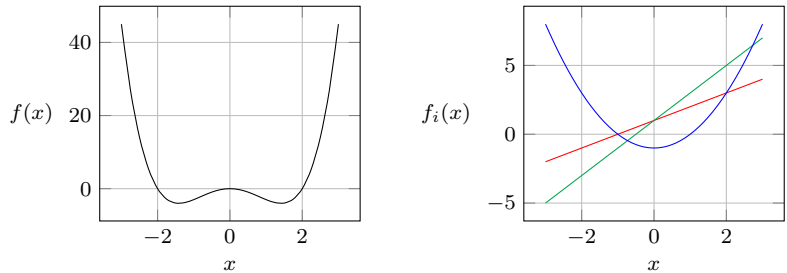
\begin{figure}
\begin{center}
  \footnotesize
  \begin{tikzpicture}
    \begin{axis}[smooth, scale=0.5, xlabel={$x$}, ylabel={$f(x)$}, ylabel style={rotate=-90}, grid]
      \addplot [domain=-3:3] { x^4 - 4*x^2 };
    \end{axis}
  \end{tikzpicture}
  \qquad
  \begin{tikzpicture}
    \begin{axis}[smooth, scale=0.5, xlabel={$x$}, ylabel={$f_i(x)$}, ylabel style={rotate=-90}, grid]
      \addplot [color=red, domain=-3:3] { x+1 };
      \addplot [color=Green, domain=-3:3] { 2*x+1 };
      \addplot [color=blue, domain=-3:3] { x^2-1 };
    \end{axis}
  \end{tikzpicture}
\end{center}
\caption{The polynomial $f =x^4 - 4x^2$ and the family $f_1 = x + 1$, $f_2 =
2x + 1$, $f_3 = x^2 - 1$ of polynomials discussed in the
Introduction\label{fig:intro}}
\end{figure}

In this paper, we are interested in the number of possible real types of
families of $n$ polynomials subject to a degree bound $d$. Note that
Eq.~\eqref{eq:real_type_1} is not a possible real type of a polynomial of degree
$3$, although it contains only three zeros. \citet[pp.11--13]{Kosta:16a} has
shown that the number of possible real types of a single polynomial of degree
$d$ is exactly
\begin{equation}
    \label{eq:kosta}
    R_d = \sum_{m=0}^d R_d^{(m)}\quad \text{with} \quad
    R_d^{(m)} = 2 \cdot \sum_{j \geq 0} \binom{m}{2m - d + 2j}
\end{equation}
and considered it unlikely that Eq.~\eqref{eq:kosta} could be expressed in a
closed form, pointing at an existing result by \citet{PWZ:96} that a sum of
the form $\sum_{j \geq l} \binom{m}{j}$ cannot be expressed in closed form as
a function of $m$ an $l$, unless one resorts to hypergeometric functions. We
are not aware of any results on the number of possible real types of families
$f_1$, \dots,~$f_n$ of polynomials for $n > 1$.

Let us make precise the notion of a closed form for our purposes here. A
\emph{closed form} represents a function using a term in arithmetic symbols
and elementary functions. Let $F_n$ denote the $n$-th Fibonacci number. Then
we allow ourselves to say that $2 F_n - 1$ is in closed form, because there
exists a closed form for $F_n$ via Binet's formula. An \emph{explicit form}
additionally admits finite sums and products with parametric upper bounds, but
not recursion.

Real types play a key role in a number of decision and quantifier elimination
procedures for real closed fields. Quantifier elimination methods based on
virtual substitution beyond degree $2$ of the quantified variables
systematically enumerate all possible real types of certain parametric
polynomials occurring in the input formula \cite{Weispfenning:97b,Sturm:18a}.
Modern implementations of this approach by Ko{\v s}ta in Redlog
\citep{Kosta:16a,DolzmannSturm:97a} and by Tonks in Maple
\citep{10.1007/978-3-030-41258-6_13,Tonks:21} have evolved from a seminal
article by \citet{Weispfenning:94a}. Virtual substitution methods are well
known for numerous applications in the sciences
\cite{DolzmannSturm:98a,Sturm:17a}. Recent application areas include chemical
reaction network theory in the sense of \citet{Feinberg:19a} in Redlog
\citep{GrigorievIosif:20a,RahkooySturm:21a} and economics in Maple
\citep{DBLP:conf/scsquare/MulliganBDET18}.

Another group of quantifier elimination procedures employs disjunctive normal
form computations and systematic case distinctions on parameters to arrive at
systems of univariate constraints such as $f_1 > 0$, $f_2 = 0$, $f_3 \leq 0$
with $f_i$ as above. The existence of a solution of the system is equivalent
to the occurrence of one of the columns $[1, 0, 0]^T$ or $[1, 0, -1]^T$ in the
real type Eq.~\eqref{eq:real_type_m}. This approach has been described by
\citet{Hormander:05a}, based on an unpublished manuscript by Paul Cohen. An
algorithm by Muchnik \citep{MICHAUXOZTURK:02a,Schoutens:04a} is closely
related. In spite of its comparatively high asymptotic complexity, the
approach has received quite some attention for its simplicity. For instance,
it was used as a basis for a fully proof-producing quantifier elimination
procedure by \citet{McLaughlinHarrison:05a}.

With both quantifier elimination approaches discussed above, closed formulas
for the number of real types help with the computation of asymptotic
complexity bounds.

We consider finite families $f_1$, \dots,~$f_n$ of polynomials, and we denote by $d_1$,
\dots,~$d_n$ their degrees and by $m$ the finite cardinality of the union of
their real roots. In terms of these notions we are going to derive here  the
following results:
\begin{enumerate}
    \item
    In the special case $n=1$, we simplify Ko\v sta's result Eq.~\eqref{eq:kosta} to
    a simple case distinction
    \begin{equation*}
        R_d = \begin{cases}
            2  F_{d+2}     & \text{if $d$ is even}, \\
            2 F_{d+2} - 2 & \text{if $d$ is odd},
            \end{cases}
    \end{equation*}
    based on Fibonacci numbers, which can be expressed in closed form, in
    contrast to sums of binomial coefficients; see Theorem~\ref{thm:T-fib}.
    \item On these grounds, we derive another closed form for the number of all real types
    realized by a single polynomial \emph{up to} degree $d$; see
    Theorem~\ref{thm:real-types-up-to-d}:
    \begin{equation*}
         \Rhat_d = 2 \cdot F_{d+3} - 2.
    \end{equation*}
    It follows that $R_d$, $\Rhat_d \in \Theta(\varphi^d)$, where $\varphi$ is
    the golden ratio; see Corollary \ref{cor:fib-asymptotic}.
    \item For the general case $n \geq 1$, we derive a
    formula for the number of all real types realized by
    polynomials with degrees $d_i$ and $m$ distinct real roots:
    \begin{equation}
    \label{eq:40inintro}
    R_{d_1, \dots, d_n}^{(m)} =
    \sum_{i=0}^m {(-1)}^i \binom{m}{i} \prod_{j=1}^n \sum_{k=0}^{m-i} \binom{m-i}{k} R^{(k)}_{d_j}
    \end{equation}
    with $R^{(k)}_{d_j}$ defined as in Eq.~\eqref{eq:kosta}. Summation over $m$
    yields an explicit form for the number  $R_{d_1,
    \dots, d_n}$ of all real types realized by  polynomials with degrees $d_i$
    and any number of roots, in analogy to Eq.~\eqref{eq:kosta}. However, compared
    to the special case $n = 1$ above, we do not obtain a closed form; see
    Theorem~\ref{thm:sign-matrix-formula}.
    \item
    For the number of all real types realized by polynomials with any choice of $d_1$, \dots,~$d_n$
    and $m$ distinct real roots, we can reduce Eq.~\eqref{eq:40inintro} to the closed form
    \begin{equation*}
        S_n^{(m)} = 2^n \cdot {(3^n - 1)}^m,
    \end{equation*}
    which generalizes the obvious equation $S_1^{(m)} = 2^{m+1}$. In
    particular, this imposes an upper bound on Eq.~\eqref{eq:40inintro}; see
    Proposition~\ref{prop:families-nice-formula}.
\end{enumerate}

Our results are applicable to the complexity analysis of algorithms involving
real types. This includes but is not limited to various approaches of real
quantifier elimination, some of which have been mentioned above.

Section~\ref{sec:preliminaries} presents relevant definitions and basic
results used throughout the article.
Section~\ref{sec:counting_real_types_of_single_polynomials} discusses the
special case of an single polynomial, and
Section~\ref{sec:counting_real_types_of_families_of_polynomials} subsequently
generalizes to families of polynomials. In a final
Section~\ref{sec:concluding_remarks_and_further_work} we finally summarize and
evaluate our results.

\section{Preliminaries}\label{sec:preliminaries}

In the following, we denote with $\Signs = \{-1, 0, 1\}$ the set of all possible signs.

\subsection{Real types of single polynomials}

Consider a non-zero univariate polynomial $f \in \R[x]$ of degree $d$ and assume $f$ has exactly $m$ real zeros
$x_1 < \cdots < x_m$. Then the real line can be partitioned into $2m + 1$ subsets 
\begin{equation*}
    \R = (-\infty, x_1) \cup \{x_1\} \cup (x_1, x_2) \cup \{x_2\} \cup \dots \cup \{x_n\} \cup (x_m, \infty) 
\end{equation*}
where the sign of $f$ is constant on each subset. The \emph{real type} of $f$ is the sign sequence $s \in {\Signs}^{2m+1}$
where $s_i \in \Signs$ is the sign of $f$ on the $i$-th subset.

\begin{example}
    The linear polynomials $2x + 1$ and $-x + 3$ have real types $[-1, 0, 1]$
    and $[1, 0, -1]$, respectively. The quadratic polynomial $x^{2} + 2x - 3$
    has real type $[1, 0, -1, 0, 1]$, and the cubic polynomial $x^3$ has again
    real type $[-1, 0, 1]$.
\end{example}

Let $s \in {\Signs}^{2m+1}$. Then wey say $s$ is a \emph{real $d$-type} if it can be realized by a polynomial of degree $d$, i.e.\ there
exists a non-zero polynomial $f \in \R[x]$ of degree $d$ with real type $s$.
For example
\begin{example}
    On the one hand, $[-1, 0, 1]$ is both a real $1$-type and a real $3$-type,
    since it can be realized by $2x + 1$ and $x^3$, respectively. On the other
    hand, $[1, 0, 1]$ is a real $2$-type, since it can be realized by $x^2$,
    but it is not a real $1$-type.
\end{example}

\citet[Section 2.1]{Kosta:16a} derived a formula for the number of real $d$-types. The main
step was to prove the following proposition, which gives a formula for the
number $R_d^{(m)}$ of real $d$-types which are realized by polynomials with exactly $m$ distinct real roots.

\begin{proposition}[Ko\v sta, 2016]\label{prop:kosta-Rdm}
\pushQED{\qed}
The number $R_d^{(m)}$ of real $d$-types realized by
polynomials with exactly $m$ distinct real roots is given by
\begin{equation*}
    R_d^{(m)} = 2 \sum_{j \geq 0} \binom{m}{2m-d+2j}.\qedhere
\end{equation*}
\popQED
\end{proposition}

Note that the sum in the proposition is finite, since the binomial
coefficients vanish for $2m-d+2j > m$.

Since a polynomial of degree $d$ has at most $d$ distinct real roots, it
follows directly that the number $R_d$ of all real $d$-types is given by the
following equation \citep[Equation 2.3]{Kosta:16a}:
\begin{equation}\label{eq:kosta23}
    R_d = \sum_{m=0}^d R_d^{(m)} =  2 \cdot \sum_{m=0}^d \sum_{j \geq 0} \binom{m}{2m - d + 2j}.
\end{equation}

Table~\ref{tab:num-rt} shows the numbers of real $d$-types up to $d = 10$
\citep[Table 2.1]{Kosta:16a}. All values in the table are even, because $f
\mapsto -f$ induces a bijection between real $d$-types realized by polynomials
with a positive and with a negative leading coefficient. Note that we do not
count the zero polynomial. It is possible to define the real type of $0
\in \R[x]$ as $[0]$, which is not realized by any non-zero polynomial. This
would increase all relevant bounds in the table  by $1$.

\begin{table}
\centering
\begin{tabular}{lrrrrrrrrrrr}
\hline
$d$    & 0 & 1 & 2 & 3 & 4 & 5 & 6 & 7 & 8 & 9 & 10 \\
$T(d)$ & 2 & 2 & 6 & 8 & 16 & 24 & 42 & 66 & 110 & 176 & 288 \\
\hline\end{tabular}
\caption{Exact numbers of real $d$-types for $d = 0, \dots, 10$}\label{tab:num-rt}
\end{table}

\subsection{The structure of real types}

We present some facts about the structure of real types of single
polynomials that will be used throughout the paper.
We begin with a key lemma by Ko\v sta that gives constraints on the entries of a real type
depending on the multiplicities of the real roots of a polynomial.

\begin{lemma}[Ko\v sta, 2016; Lemma 4]\label{lem:kosta-mult}
\pushQED{\qed}
Let $[s_1, \dots, s_{2m+1}]$ be a real type of a positive-degree polynomial $f
\in \R[x]$ with $m$ distinct real roots $\beta_1 < \cdots < \beta_m$ with
multiplicities $k_1$, \dots,~$k_m$, respectively. Then for each $j \in \{1,
\dots, m\}$, we have
\begin{equation*}
    s_{2j+1} = {(-1)}^{k_j} \cdot s_{2j-1}.\qedhere
\end{equation*}
\popQED
\end{lemma}

As a consequence, the real type of a polynomial depends only on the sign of
its leading coefficient and the parity of the multiplicities of its real roots.

\begin{lemma}\label{le:real-type-dep}
Let $f \in \R[x]$ be a non-zero polynomial with leading coefficient $\alpha$
and $m$ real roots $\beta_1 < \cdots < \beta_m$ with multiplicities
$k_1$, \dots,~$k_m$, respectively. Then the real type of $f$ depends only on
$\sign(\alpha)$ and $k_i \Mod 2$ for $1 \leq i \leq m$.
\end{lemma}
\begin{proof}
    Denote with $s = (s_1, \dots, s_{2m+1}) \in \Signs^{2m+1}$ the real type of $f$.
By assumption, we can factor $f$ as
\begin{equation*}
    f(x) = {(x - \beta_1)}^{k_1} \cdots {(x - \beta_m)}^{k_m} \cdot g(x)
\end{equation*}
for a polynomial $g \in \R[x]$ with $g(x) \neq 0$ for all $x \in \R$. Then
\begin{align*}
s_1 = \lim_{x \to -\infty} \sign f(x) &= \sign(\alpha) \cdot {(-1)}^{\deg f} \\ &= \sign(\alpha) \cdot {(-1)}^{\deg f \Mod 2}
\end{align*}
where $\deg f = \sum_{i=1}^m k_i + \deg g$. Since $g$ has no real roots, the intermediate value theorem implies that $\deg g$ is even.
Hence,
\begin{equation*}
    s_1 = \sign(\alpha) \cdot {(-1)}^{\sum_{i=1}^m k_i \Mod 2}
\end{equation*}
which depends only on $\sign(\alpha)$ and $k_i \Mod 2$.
Lemma~\ref{lem:kosta-mult} then implies that
\begin{equation*}
    s_{2j+1} = \sign(\alpha) \cdot {(-1)}^{\sum_{i=j+1}^m k_i \Mod 2}
\end{equation*}
for all $1 \leq j \leq m$. Additionally, we always have $s_{2j} = 0$ for all $1 \leq j \leq m$.
Thus, the real type $s$ depends only on $\sign(\alpha)$ and $k_i \Mod 2$.
\end{proof}

Using this lemma, Ko\v sta arrives at Proposition~\ref{prop:kosta-Rdm} by
counting all possible parities of multiplicities for a given degree $d$ and a
given number of real roots $m$. Moreover, it follows that a sign sequence is a real
type if and only if it is a real type of a polynomial with integer
coefficients. This is important for algorithmic purposes, where input over a
finite alphabet is required.

\begin{proposition}\label{prop:z-and-r}
A sign sequence in $\Signs^{2m+1}$ is the real type of a polynomial in $\R[x]$
if and only if it is the real type of a polynomial in $\Z[x]$.
\end{proposition}
\begin{proof}
Let $s \in \Signs^{2m+1}$ be the real type of a polynomial $f \in \R[x]$.
Assume that $f$ has leading coefficient $\alpha$ and $m$ real roots
$\beta_1 < \cdots < \beta_m$ with multiplicities $k_1$, \dots,~$k_m$,
respectively. Define the polynomial
\begin{equation*}
    g(x) = \sign(\alpha) \cdot {(x - 1)}^{k_1} \cdots {(x - m)}^{k_m} \in \Z[x].
\end{equation*}
Then Lemma~\ref{le:real-type-dep} implies that $g$ has also real type $s$,
since the leading coefficients of $f$ and $g$ have the same sign and the
multiplicities of both polynomials agree. The converse follows from $\Z[x]
\subseteq \R[x]$.
\end{proof}

\subsection{Real types of families of polynomials}\label{sec:comb-sign-matrix}

We generalize our notion of real types to families of polynomials.
Let $f_1$, \dots,~$f_n \in \R[x]$ be non-zero polynomials and consider their
combined roots
\begin{equation*}
    V = \{ x \in \R \mid \text{$f_i(x) = 0$ for some $1 \leq i \leq n$} \}.
\end{equation*}
Since each polynomial has a finite number of roots, $V$ is again finite and consists
of $m$ points $x_1 < \cdots < x_m$. As in the case of a single polynomial, these points partition the real line into $2m + 1$ subsets
\begin{equation*}
    \R = (-\infty, x_1) \cup \{x_1\} \cup (x_1, x_2) \cup \{x_2\} \cup \dots \cup \{x_n\} \cup (x_m, \infty),
\end{equation*}
where the sign of each polynomial $f_i$ is constant on each subset.
Denote with $\sigma_{i,j} \in \Signs$ the sign of $f_i$ on the $j$-th subset.
Then the \emph{real type} of $f_1$, \dots,~$f_n$ is given by
the sign matrix
\begin{equation*}
\begin{bmatrix}
    \sigma_{1,1} & \hdots & \sigma_{1,2m+1} \\
    \vdots & \ddots & \vdots \\
    \sigma_{n,1} & \hdots & \sigma_{n,2m+1}
\end{bmatrix}
\in \Signs^{n \times (2m+1)}.
\end{equation*}

It follows that every odd column of a real type contains exclusively $-1$ or $1$, and
every even column contains at least one zero. Furthermore, the $i$-th row is given
by the real type of $f_i$, possibly with some entries repeated.

It is not hard to see that Proposition~\ref{prop:z-and-r} generalizes to
families of polynomials and sign matrices.

\section{Counting real types of single polynomials}\label{sec:counting_real_types_of_single_polynomials}
We want to improve the explicit form given in Eq.~\eqref{eq:kosta23} to a closed form, which
will be based on \emph{Fibonacci numbers} $F_n$. The Fibonacci sequence is recursively defined
as $F_1 = 1$, $F_2 = 1$ and $F_n = F_{n-1} + F_{n-2}$ for $n \geq 3$. It starts with
\begin{equation}\label{eq:fib-seq}
    1,\quad 1,\quad 2,\quad 3,\quad 5,\quad 8,\quad 13,\quad 21,\quad \dots
\end{equation}

Before we can prove our first theorem, we first need the following two lemmas,
which establish a recursive relation and base cases for the number $R_d^{(m)}$
of real $d$-types realized by polynomials with exactly $m$ distinct
real roots.

\begin{lemma}\label{lem:R-even-odd}
Let $d \geq 0$. Then
\begin{equation*}
    R_d^{(0)} = \begin{cases}
        2 & \text{if $d$ is even}, \\
        0 & \text{if $d$ is odd}.
    \end{cases}
\end{equation*}
\end{lemma}
\begin{proof}
First, assume that the degree $d$ is even. Then, the intermediate value theorem implies that
any real $d$-type without zeros is of the form $(1, \dots, 1)$ or $(-1, \dots, -1)$.
Both real types can for example be realized by $x^d + 1$ and $-x^d - 1$, respectively,
which shows $R_d^{(0)} = 2$.
Next, assume that $d$ is odd. By the intermediate value theorem, every polynomial of odd degree has at least one
real zero. Thus, $R_d^{(0)} = 0$ in this case.
\end{proof}

\begin{lemma}\label{lem:R-rec}
Let $d \geq 2$ and $m \geq 1$. Then
\begin{equation*}
    R_d^{(m)} = R_{d-2}^{(m-1)} + R_{d-1}^{(m-1)}.
\end{equation*}
\end{lemma}
\begin{proof}
Let $j \geq 0$. Then the relation $\binom{n}{k} = \binom{n-1}{k} + \binom{n-1}{k-1}$
for all $n \geq 1$ and $k \in \mathbb{Z}$ yields
\begin{align*}
    \binom{n}{2n-d+2j}
    & = \binom{n-1}{2n-d+2j} + \binom{n-1}{2n-d+2j-1}\\
    & = \binom{n-1}{2(n-1)-(d-2)+2j} + \binom{n-1}{2(n-1)-(d-1)+2j}.
\end{align*}
Hence, by summing over all $j \geq 0$ on both sides and using Proposition~\ref{prop:kosta-Rdm}, we obtain
\begin{equation*}
    R_d^{(m)} = R_{d-2}^{(m-1)} + R_{d-1}^{(m-1)}.\qedhere
\end{equation*}
\end{proof}

Using the previous two lemmas, we can show by induction that the number of
real $d$-types is given by a case distinction of closed forms involving
Fibonacci numbers.

\begin{theorem}\label{thm:T-fib}
The number $R_d$ of real $d$-types is
\begin{equation*}
    R_d = \begin{cases}
        2 F_{d+2}     & \text{if $d$ is even}, \\
        2 F_{d+2} - 2 & \text{if $d$ is odd}.
    \end{cases}
\end{equation*}
\end{theorem}
\begin{proof}
We prove the statement by induction on $d$. The cases $d = 0$ and $d = 1$ can
directly be checked using Table~\ref{tab:num-rt} and Eq.~\eqref{eq:fib-seq}.
Now, let $d \geq 2$ and assume the statement holds for $d-1$ and $d-2$. Then
Lemma~\ref{lem:R-rec} yields
\begin{equation*}
    R_d = \sum_{m=0}^d R_{d}^{(m)} =
    R_{d}^{(0)} + \sum_{m=1}^{d} \big(R_{d-1}^{(m-1)} + R_{d-2}^{(m-1)}\big).
\end{equation*}
By splitting up the sum and performing an index shift, we obtain
\begin{equation*}
    R_d = R_{d}^{(0)} + \sum_{m=0}^{d-1} R_{d-1}^{(m)} + \sum_{m=0}^{d-2} R_{d-2}^{(m)} + R_{d-2}^{(d-1)}.
\end{equation*}
Furthermore, a polynomial of degree $d-2$ has at most $d-2$ distinct real roots, which implies $R_{d-2}^{(d-1)} = 0$.
Thus,
\begin{equation*}
    R_d = R_{d}^{(n)} + R_{d-1} + R_{d-2}.
\end{equation*}
If $d$ is even, Lemma~\ref{lem:R-even-odd} and the induction hypothesis yield
\begin{equation*}
    R_d = 2 + (2F_{d+1} - 2) + 2F_{d} = 2F_{d+2}.
\end{equation*}
Otherwise, $d$ is odd and
\begin{equation*}
    R_d = 2F_{d+1} + (2F_{d} - 2) = 2F_{d+2} - 2.\qedhere
\end{equation*}
\end{proof}

Next, we extend the previous theorem to counting the number of real types for
polynomials up to degree $d$ instead of exactly degree $d$. The following two
lemmas show that this number is given by $R_d + R_{d-1}$. Note that we do not
sum over all smaller degrees.

\begin{lemma}\label{le:inlcusion-plus-2}
If $s$ is a real $d$-type, then $s$ is also a real $(d+2)$-type.
\end{lemma}
\begin{proof}
Let $s$ be a real $d$-type. Then $s$ is realized by a polynomial $f \in \R[x]$
of degree $d$. Define the polynomial $g = f \cdot (x^2 + 1)$.
For all $x \in \R$, we have $x^2 + 1 > 0$, and thus
\begin{equation*}
    \sign g(x) = \sign f(x).
\end{equation*}
Thus $s$ is also realized by the polynomial $g$, which has degree $d+2$.
\end{proof}

\begin{lemma}\label{le:even-odd-disjoint}
The sets of real types realized by polynomials with an even degree and with an odd degree
are disjoint.
\end{lemma}
\begin{proof}
Let $s \in \Signs^n$ be a sign sequence which is realized by a polynomial $f \in \R[x]$ of degree $d$.
Assume $d$ is even. Then
\begin{equation*}
    \lim_{x \to -\infty}\sign f(x) = \lim_{x \to \infty}\sign f(x),
\end{equation*}
which implies $s_1 = s_n$. On the other hand, if $d$ is odd, then
\begin{equation*}
    \lim_{x \to -\infty}\sign f(x) = -\lim_{x \to \infty}\sign f(x),
\end{equation*}
which implies $s_1 = -s_n$. Thus, a sign sequence $s$ can not be realized by a polynomial of even degree and
by a polynomial of odd degree at the same time.
\end{proof}

After these preparations, we can state our theorem.

\begin{theorem}\label{thm:real-types-up-to-d}
Then the number $\Rhat_d$ of real types realized by polynomials up to
degree $d$ is $\Rhat_d = 2 F_{d+3} - 2$.
\end{theorem}
\begin{proof}
In the case $d = 0$, Table~\ref{tab:num-rt} and Eq.~\eqref{eq:fib-seq} yield $\Rhat_0 = 2 = 2 F_3 - 2$.
Now assume $d \geq 1$. According to Lemma~\ref{le:inlcusion-plus-2} and
Lemma~\ref{le:even-odd-disjoint}, any real type which can be realized by a
polynomial up to degree $d$ can either be realized by a polynomial of degree
$d$ or by a polynomial of degree $d-1$. Thus, $\Rhat_d = R_d + R_{d-1}$.
Applying Theorem~\ref{thm:T-fib}
yields that this number is given by
\begin{equation*}
    \Rhat_d = R_d + R_{d-1} = 2 \cdot F_{d+2} + 2 \cdot F_{d+1} - 2 = 2 \cdot F_{d + 3} - 2.\qedhere
\end{equation*}
\end{proof}

Given Theorem~\ref{thm:T-fib} and Theorem~\ref{thm:real-types-up-to-d},
Binet's formula can be used to obtain a closed formula for the numbers $R_d$
and $\Rhat_d$.  This additionally allows to determine their precise
asymptotic growth.

\begin{proposition}[Binet's formula]\label{prop:fib-closed-form}
The Fibonacci numbers can be expressed in closed form as
\begin{equation*}
    F_{n} = \frac{1}{\sqrt{5}}\big(\varphi^{n} - {\varphi}^{-n}\big),
\end{equation*}
where $\varphi = \frac{1 + \sqrt{5}}{2} \approx 1.618$ is the golden ratio.\qed
\end{proposition}

\begin{corollary}\label{cor:fib-asymptotic}
Let $R_d$ be the number of real types realized by polynomials of exactly degree $d$, let $\Rhat_d$ be the number of real types realized by polynomials up to degree $d$, and let $\varphi$ denote the golden ratio. Then the following holds:
\begin{equation*}
R_d = \Theta(\varphi^d),\quad \Rhat_d = \Theta(\varphi^d).
\end{equation*}
\end{corollary}

\begin{proof}
This follows immediately from Theorem~\ref{thm:T-fib},
Theorem~\ref{thm:real-types-up-to-d} and
Proposition~\ref{prop:fib-closed-form} since $|\varphi^{-1}| < 1$, which
implies ${\varphi}^{-n} \to 0$ as $n \to \infty$.
\end{proof}

\section{Counting real types of families of polynomials}\label{sec:counting_real_types_of_families_of_polynomials}

We are interested in two possible questions regarding sign matrices of
families of real polynomials, where (\ref{item:a}) is a generalization of Ko\v
sta's result from Proposition~\ref{prop:kosta-Rdm} for a single polynomial:
\begin{enumerate}[label=(\alph*)]
\item\label{item:a} the number of sign matrices as a function of matrix dimensions and polynomial degrees;
\item the number of sign matrices as a function of matrix dimensions, for
arbitrary polynomial degrees.
\end{enumerate}

\subsection{Sign matrices depending on degrees}

Consider $n$, $m \in \N$ and degrees $d_1$, \dots,~$d_n \geq 1$. We want to
count the number $R_{d_1, \dots, d_n}^{(m)}$ of matrices in $\Signs^{n \times
(2m+1)}$ that establish a real $(d_1, \dots, d_n)$-types for a family of $n$
polynomials.

As a first step, we characterize these matrices as follows.

\begin{proposition}\label{prop:comb-sign-char}
A matrix in $\Signs^{n \times (2m+1)}$
is a real $(d_1, \dots, d_n)$-type if and only if the following hold:
\begin{enumerate}
\item Every even column contains at least one zero.
\item Every odd column does not contain any zero.
\item Row $i$ is a real $d_i$-type up to repeated entries.
\end{enumerate}
\end{proposition}

\begin{proof}
Let $A \in \Signs^{n \times (2m+1)}$ be a real $(d_1, \dots, d_n)$-type. Then
all three properties follows directly from the definition of real $(d_1,
\dots, d_n)$-types in Section~\ref{sec:comb-sign-matrix}. Conversely, consider
 $A \in \Signs^{n \times (2m+1)}$ that satisfies the three properties.
According to property (3), each row $i$ of $A$ is the real type of a
polynomial $f_i$ up to repeated entries. Using the construction in the proof
of Proposition~\ref{prop:z-and-r}, we may assume without loss of generality
that $f_i$ has its real roots at the points $j \in \{1, \dots, 2m+1\}$ with
$A_{i,j} = 0$. It is not hard to see that the real type of the polynomials
$f_1$, \dots,~$f_n$ thus obtained is our matrix $A$. Hence, $A$ is a real $(d_1,
\dots, d_n)$-type.
\end{proof}

Since we know how to compute the number of real $d_i$-types, we can focus on
the distribution of zeros in sign matrices. To handle the condition that every
even column contains at least one zero, we use an inclusion-exclusion argument
in the sense of \citet[Theorem 10.1]{vanLint:92}:

\begin{proposition}[van Lint and Wilson, 1992]\label{prop:incl-excl}
\pushQED{\qed}
Let $S$ be a set of size $N$ and $E_1$, \dots,~$E_r$ not necessarily
distinct subsets of $S$. For any subset $M$ of $\{1, \dots , r\}$, 
we define $N(M)$ to be the number of elements of $S$ in 
$\bigcap_{i \in M} E_i$ and for $0 \leq j \leq r$, we define $N_j = \sum_{|M|=j} N(M)$.
Then the number of elements of $S$ not in any of the subsets $E_i$, $1 \leq i \leq r$, is 
\begin{equation*}
    N - N_1 + N_2 - N_3 + \cdots + {(-1)}^r N_r.\qedhere
\end{equation*}
\popQED
\end{proposition}

In order to apply this proposition to real types, we define
the set $S$ as the set of all matrices in $\Signs^{n\times(2m+1)}$ such that
\begin{enumerate}
\item every odd column does not contain any zero;
\item row $i$ is a real $d_i$-type up to repeated entries.
\end{enumerate}
Furthermore, we define the subsets $E_{i}$ for $1 \leq i \leq m$ as the sets
of all elements in $S$ that satisfy the additional constraint
\begin{enumerate}
\item[3.] column $2i$ does not contain a zero.
\end{enumerate}
On these grounds, Proposition~\ref{prop:comb-sign-char} yields that the number
$R_{d_1,\dots, d_n}^{(m)}$ of real $(d_1, \dots, d_n)$-types with $m$ distinct
real roots is given by
\begin{equation*}
    R_{d_1,\dots, d_n}^{(m)} = |S \setminus (E_1 \cap \dots \cap E_m)|,
\end{equation*}
which can, in turn, be computed using the inclusion-exclusion formula from
Proposition~\ref{prop:incl-excl}. However, before we can apply the formula, we
first have to determine the numbers $N(M)$ of elements in the intersections
$\bigcap_{i \in M} E_i$ for subsets $M \subseteq \{1, \dots, m\}$.

\begin{lemma}\label{lem:N(M)}
Let $M \subseteq \{1, \dots, m\}$ be a possibly empty subset.
Then the number $N(M)$ of matrices in $S$ that contain no zero in column
$2j$ for all $j \in M$ is given by 
\begin{equation*}
    N(M) = \prod_{i=1}^n \sum_{k=0}^{m-|M|} \binom{m-|M|}{k} R_{d_i}^{(k)}.
\end{equation*}
\end{lemma}
\begin{proof}
Consider the row $i$ of a matrix in $S$ and assume it contains exactly 
$k$ zeros in the columns $\{2, 4, \dots, 2m\} \setminus 2M$.
Then there are $\binom{m-|M|}{k}$ possibilities for choosing these zeros,
and for each selection of zeros there are $R_{d_i}^{(k)}$ possibilities to choose the remaining entries 
from $\{-1, 1\}$. 
Thus, by summing over all possible values of $k$, we obtain
\begin{equation*}
    \sum_{k=0}^{m-|M|} \binom{m-|M|}{k} R_{d_i}^{(k)}
\end{equation*}
possibilities for row $i$. 
Since the entries in each of the $n$ rows do not depend on the entries of other rows, the total number 
of matrices is given by
\begin{equation*}
    N(M) = \prod_{i=1}^n \sum_{k=0}^{m-|M|} \binom{m-|M|}{k} R_{d_i}^{(k)}.\qedhere
\end{equation*}
\end{proof}

Using this lemma, we can apply the inclusion-exclusion formula to obtain a
formula for the number of real types of families of polynomials.

\begin{theorem}\label{thm:sign-matrix-formula}
The number $R_{d_1,\dots,d_n}^{(m)}$ of matrices in $\Signs^{n \times (2m+1)}$ that are
real $(d_1,\dots,d_n)$-types is given by
\begin{equation}\label{eq:fortythree}
    R_{d_1,\dots,d_n}^{(m)} = \sum_{i=0}^m {(-1)}^i \binom{m}{i} \prod_{j=1}^n \sum_{k=0}^{m-i} \binom{m-i}{k} R_{d_j}^{(k)}.
\end{equation}
\end{theorem}
\begin{proof}
Using the previous notation, Proposition~\ref{prop:comb-sign-char} and Proposition~\ref{prop:incl-excl} imply that
the number of real types is given by
\begin{equation}\label{eq:inc-exc-formula}
    R_{d_1,\dots,d_n}^{(m)} 
    = |S \setminus (E_1 \cap \dots \cap E_m)|
    = \sum_{i=0}^m {(-1)}^i N_i,
\end{equation}
where $N_i = \sum_{|M| = i} N(M)$ and $N(M)$ is the number of matrices in $S$
that for all $j \in M$ there occurs no zero in column $2j$.
By Lemma~\ref{lem:N(M)}, $N(M)$ depends only on $|M|$ and is given by
\begin{equation*}
    N(M) = \prod_{j=1}^n \sum_{k=0}^{m-|M|} \binom{m-|M|}{k} R_{d_j}^{(k)}.
\end{equation*}
Furthermore, there are $\binom{m}{i}$ possibilities of
choosing $M \subseteq \{1, \dots, m\}$ with $|M| = i$, which yields
\begin{equation*}
    N_i = \binom{m}{i} \prod_{j=1}^n \sum_{k=0}^{m-i} \binom{m-i}{k} R_{d_j}^{(k)}.
\end{equation*}
Substituting $N_i$ back into Eq.~\eqref{eq:inc-exc-formula} shows
that the number of real types is given by 
\begin{equation*}
    R_{d_1,\dots,d_n}^{(m)} = \sum_{i=0}^m {(-1)}^{i} \binom{m}{i} \prod_{j=1}^n \sum_{k=0}^{m-i} \binom{m-i}{k} R_{d_j}^{(k)}.\qedhere
\end{equation*}
\end{proof}

As with a single polynomial, we can eliminate the dependence on $m$ and count sign matrices with an arbitrary number of columns by summing over all possible numbers of roots.

\begin{corollary}\label{cor:Rd1dn}
The number $R_{d_1,\dots,d_n}$ of sign matrices that are real $(d_1,\dots,d_n)$-types is given by
\begin{equation}\label{eq:Rd1dn}
    R_{d_1,\dots,d_n} = \sum_{m = 0}^{d_1 + \dots + d_n} R^{(m)}_{d_1,\dots,d_n}.
\end{equation}
\end{corollary}
\begin{proof}
The statement follows directly from Theorem~\ref{thm:sign-matrix-formula}
since sets of sign matrices of different dimensions are disjoint, and $n$
non-zero polynomials of degrees $d_1$, \dots,~$d_n$ have at most $d_1 + \dots
+ d_n$ many roots.
\end{proof}

Additionally, we can regard the degrees $d_1$, \dots,~$d_n$ as only upper bounds
and obtain a generalization of Theorem~\ref{thm:real-types-up-to-d} for families of polynomials.

\begin{corollary}
The number $\Rhat_{d_1,\dots,d_n}$ of sign matrices realized by 
$n$ non-zero polynomials up to degrees $d_1$, \dots,~$d_n$ is given by
\[
    \Rhat_{d_1,\dots,d_n} = \sum_{\delta \in \{0, 1\}^n} R_{d_1 - \delta_1, \dots, d_n-\delta_n}.
\]
\end{corollary}
\begin{proof}
As in the proof of Theorem~\ref{thm:sign-matrix-formula}, we have $\Rhat_d = R_d + R_{d-1}$ for a single polynomial. Since the rows of a sign matrices are given by $n$ independent real types of single polynomials, the statement follows.
\end{proof}
    
Note that the formulas for $R_{d_1, \dots, d_n}$ and $\Rhat_{d_1, \dots, d_n}$
in the previous two corollaries exclude the zero polynomial. It is possible to
adjust the formulas to cover the zero polynomial via summation over all
sublists of $(d_1, \dots, d_n)$.

\begin{example}\label{ex:hong}
    Hong discussed the following example as an application of his method of
    slope resultants \cite[Section 5]{Hong:93b}:
    \begin{equation*}
        \exists x(f_1 = 0 \land f_2 = 0 \land f_3 \leq 0),
    \end{equation*}
    where $f_1 = ux^2+vx+1$, $f_2 = vx^3+wx+u$, and $f_3 = wx^2+vx+u$. This
    asks for a Boolean combination of polynomial constraints in the parameters
    $u$, $v$, $w$ for the system $f_1=0$, $f_2=0$, $f_3 \geq 0$ to have a real
    solution. Given a choice of parameters, the answer is positive if and only
    if the corresponding sign matrix of the family $f_1$, $f_2$, $f_3$ with
    degrees $2, 3, 2$, respectively, contains a column $[0, 0, 0]^T$ or a
    column $[0, 0, 1]^T$. It follows from Lemma~\ref{le:real-type-dep} that
    only finitely many sign matrices need be considered. Those sign matrices
    can be enumerated via systematic case distinctions on the vanishing of
    parametric coefficients, possibly enumerating sign matrices multiple times
    \cite{Hormander:05a,Schoutens:04a}. We are interested in an upper bound on
    the number of sign matrices thus enumerated.

    The number of $R_{2, 3, 2}$ of real $(2, 3, 2)$-types can be computed via
    Eq.~\eqref{eq:Rd1dn} in Corollary~\ref{cor:Rd1dn}. With our application
    here, we enumerate all possible combinations of non-vanishing leading
    terms of $f_1$, $f_2$, $f_3$, where we assume for simplicity that all
    coefficients are parametric and algebraically independent. The following
    upper bound on the number of sign matrices enumerated takes into
    consideration multiple occurrences as discussed above:
    \begin{equation}
        \bar{R}_{d_1, d_2, d_3} =
            \sum_{e_1 = 0}^{d_1} \sum_{e_2 = 0}^{d_2} \sum_{e_3 = 0}^{d_3} R_{e_1, e_2, e_3}.
    \end{equation}
    Note that this equation assumes that none of the $f_1$, $f_2$, $f_3$
    vanishes altogether, which is owed to the general exclusion of zero
    polynomials throughout our work here. The possible vanishing of one or
    several of the $f_1$, $f_2$, $f_3$ is finally modeled by considering all
    sublists of $(d_1, d_2, d_3)$:
    \begin{equation}
        \bar{\bar{R}}_{d_1, d_2, d_3} = \sum_{e \subseteq (d_1, d_2, d_3)} \bar{R}_e.
    \end{equation}
    Using the Python code in Figure~\ref{fig:hong}, we quickly compute $R_{2, 3,
    2} = 26624$, $\bar{R}_{2, 3, 2} = 53736$, and our main result
    \begin{equation}\label{eq:hong_final}
        \bar{\bar{R}}_{2, 3, 2} = 55339 \approx 5 \cdot 10^4.
    \end{equation}
\end{example}

\begin{example}\label{ex:weispfen}
    Similarly to Example~\ref{ex:hong},
    Weispfenning introduced the following first-order formula as a benchmark for real
      quantifier elimination \cite[Section 5.5]{Weispfenning:97b}:
    \begin{equation*}
        \exists x\left(\bigwedge_{i=1}^2 f_i = 0 \land \bigwedge_{i=3}^5 f_i < 0\right),
        \quad
        f_i = a_i x^2 + b_i x + c_i.
    \end{equation*}
    Here, the answer is positive for choices of parameters $a_1$, \dots,~$c_5$
    for which the sign matrix of the family $f_1$, \dots,~$f_5$ contains a column
    $[0, 0, -1, -1, -1]^T$. Using our Python code in Figure~\ref{fig:hong} once again, we obtain:
        \begin{equation}\label{eq:weispfen_final}
        \bar{\bar{R}}_{2, 2, 2, 2, 2} = 311476091 \approx 3 \cdot 10^8.
    \end{equation}
    The large factor between the upper bounds on the numbers of relevant sign
    matrices in Eq.~\eqref{eq:weispfen_final} and Eq.~\eqref{eq:hong_final}
    supports the intuition that Weispfenning's quantifier elimination example
    is harder than Hong's.
\end{example}

\begin{figure}
\begin{lstlisting}[language=python, commentstyle=\rmfamily\itshape,%
        basicstyle=\ttfamily\scriptsize, frame=single]
from itertools import chain, combinations, product as cart_prod
from math import comb, prod
from typing import Iterator

def binomial(n: int, k: int) -> int:
    # Binomial coefficients with possibly negative arguments
    return comb(n, k) if k >= 0 else 0

def all_sublists(L: list[int]) -> Iterator[list[int]]:
    itr = chain(*(combinations(L, i) for i in range(len(L) + 1)))
    for tup in itr:
        yield list(tup)

def R_m(d: list[int], m: int) -> int:
    # Implements Eq. (7) in Theorem 18
    n = len(d)
    s = sum((-1)**i * binomial(m, i)
            * prod(sum(binomial(m - i, k)
                       * sum(binomial(k, 2 * k - d[j] + 2 * ll)
                             for ll in range((d[j] - k) // 2 + 2))
                       for k in range(m - i + 1))
                   for j in range(n))
            for i in range(m + 1))
    return 2 ** n * s

def R(d: list[int]) -> int:
    # Implements Eq. (9) in Corollary 19
    return sum(R_m(d, m) for m in range(sum(d) + 1))

def R_bar(d: list[int]) -> int:
    # Implements Eq. (10) in Example 21
    ranges = [range(d_i + 1) for d_i in d]
    return sum(R(list(e)) for e in cart_prod(*ranges))

def R_doublebar(d: list[int]) -> int:
    # Implements Eq. (11) in Example 21
    return sum(R_bar(dd) for dd in all_sublists(d))
\end{lstlisting}
    \caption{Python code for the computations in Example~\ref{ex:hong}\label{fig:hong}}
\end{figure}

\subsection{Sign matrices independent of degrees}

We still consider $n$, $m \in \N$, but we are now interested in the number
$S_n^{(m)}$ of sign matrices in $\Signs^{n\times(2m+1)}$ that are real types
of $n$ polynomials with arbitrary degrees. In this setting, it turns out that,
the explicit form from the previous section simplifies to a closed form. In
particular, it generalizes the following result for all real types of length
$m + 1$ in the case of a single polynomial.

\begin{proposition}\label{prop:real-type-arbit-deg}
There are $2^{m+1}$ real types realized by single polynomials with exactly
$m$ distinct real roots and arbitrary degree.
\end{proposition}
\begin{proof}
Let $s$ be a real type of a polynomial with $m$ distinct real roots. Then $s
\in \Signs^{2m+1}$ and $s_{2i} = 0$ for $1 \leq i \leq m$. Thus there are at
most $2^{m+1}$ possibilities for choosing $s_{2i+1} \in \{-1, 1\}$ for $0 \leq
i \leq m$. Conversely, the number of real types with $m$ roots is bounded from below
by $R_{2m+1}^{(m)} + R_{2m}^{(m)}$ since real types with even and odd degree
are disjoint by Lemma~\ref{le:even-odd-disjoint}.
Proposition~\ref{prop:kosta-Rdm} then implies
\begin{equation*}
    R_{2m+1}^{(m)} + R_{2m}^{(m)}
     = 2 \sum_{j \geq 0} \binom{m}{2j} + 2 \sum_{j \geq 0} \binom{m}{2j+1}\\
     = 2 \sum_{k \geq 0} \binom{m}{k} = 2 \cdot 2^m.
\end{equation*}
This agrees with the above upper bound such that the
number of real types is exactly $2^{m+1}$.
\end{proof}

Given our setting and using this proposition, we can simplify the right hand side of
Eq.~\eqref{eq:fortythree} in Theorem~\ref{thm:sign-matrix-formula}.

\begin{proposition}\label{prop:families-nice-formula}
The number $S_n^{(m)}$ of sign matrices in $\Signs^{n \times (2m+1)}$
which are real types of $n$ polynomials of arbitrary degree is given by 
\[
    S_n^{(m)} = 2^n \cdot {(3^n -1)}^m.
\]
\end{proposition}
\begin{proof}
Proposition~\ref{prop:real-type-arbit-deg} and the proof of Lemma~\ref{lem:N(M)} and Theorem~\ref{thm:sign-matrix-formula} show that the number of
real types realized by polynomials with an arbitrary degree can be obtained by replacing the terms $R_{d_j}^{(k)}$ in 
\begin{equation*}
    \sum_{i=0}^m {(-1)}^i \binom{m}{i} \prod_{j=1}^n \sum_{k=0}^{m-i} \binom{m-i}{k} R_{d_j}^{(k)}
\end{equation*}
with the terms $2^{k+1}$. By the binomial theorem, we have
\begin{equation*}
    \sum_{k=0}^{m-i} \binom{m-i}{k} 2^{k+1}
    = 2 \sum_{k=0}^{m-i} \binom{m-i}{k} 2^{k}{(1)}^{(m-i)-k}
    = 2 \cdot {(2 + 1)}^{m-i}, 
\end{equation*}
which yields
\begin{equation*}
    \prod_{j=1}^n 2 \cdot {(2 + 1)}^{m-i} 
    = 2^n \cdot 3^{(m - i) \cdot n}.
\end{equation*}
Applying the binomial theorem once again shows that the number $S_n^{(m)}$ is given by
\begin{align*}
    S_n^{(m)} &= \sum_{i=0}^m {(-1)}^i \binom{m}{i} 2^n \cdot 3^{(m - i) \cdot n} \notag\\
    &= 2^n \sum_{i=0}^m \binom{m}{i} {(-1)}^i  {(3^{n})}^{m - i} \notag\\
    &= 2^n \cdot {(3^n - 1)}^m.\qedhere
\end{align*}
\end{proof}
This generalizes the obvious equation $S_1^{(m)} = 2^{m+1}$.

\section{Concluding remarks}\label{sec:concluding_remarks_and_further_work}
We have derived explicit or even closed formulas for numbers of real types of
single polynomials as well as finite families of polynomials, subject to
various parameterizations. In particular, we have provided an explicit formula
for the number of real types of families of $n$ polynomials with degrees
$d_1$, \dots,~$d_n$ and exactly $m$ distinct real roots. For the special case
$n = 1$ we obtain even closed forms, with a case distinction depending on the
parity of the degree, which refutes a conjecture in the literature.

It remains a challenging open question whether or not closed forms exist in
the general case.

\bibliographystyle{plainnat}
\bibliography{article}

\end{document}